\documentclass[11pt]{article}

\usepackage[utf8]{inputenc}
\usepackage[T1]{fontenc}
\usepackage{amsthm}
\usepackage{amssymb}
\usepackage{amsmath}
\usepackage{enumerate}
\usepackage{subcaption}
\usepackage{hyperref}
\hypersetup{
  colorlinks   = true, 
  urlcolor     = blue, 
  linkcolor    = blue, 
  citecolor   = red    
}

\usepackage{tkz-graph}
\usepackage{tikz}
\usetikzlibrary{decorations.pathreplacing}

\usepackage{algorithm}
\usepackage{algpseudocode}

\tikzstyle{b_vertex}=[circle,fill=black!100,text=white,inner sep=0.8mm,draw]
\tikzstyle{w_vertex}=[circle,fill=white!100,text=white,inner sep=0.8mm,draw]
\tikzstyle{point}=[circle,fill=black,inner sep=0.1mm]
\tikzstyle{path_edge}=[thick]

\newtheorem{theorem}{Theorem}
\newtheorem{lemma}[theorem]{Lemma}
\newtheorem{corollary}{Corollary}

\title{Combinatorics and algorithms for augmenting graphs}

\author{Konrad K. Dabrowski\thanks{School of Engineering and  Computing Sciences, Durham University, Science Laboratories, South Road,
Durham DH1 3LE, UK. Email: konrad.dabrowski@durham.ac.uk}
\and Dominique de Werra\thanks{Mathematics Institute, \'Ecole Polytechnique F\'ed\'erale (EPFL), Switzerland. Email: dominique.dewerra@epfl.ch}
\and Vadim V. Lozin\thanks{DIMAP and Mathematics Institute, University of Warwick, Coventry, CV4 7AL, UK. Email: V.Lozin@warwick.ac.uk} 
\and Viktor Zamaraev\thanks{DIMAP and Mathematics Institute, University of Warwick, Coventry, CV4 7AL, UK. Email: V.Zamaraev@warwick.ac.uk}}

\begin{document}
\maketitle

\begin{abstract}
The notion of augmenting graphs generalizes Berge's idea of augmenting chains,  
which was used by Edmonds in his celebrated solution of the maximum matching problem. 
This problem is a special case of the more general maximum independent set (MIS) problem.
Recently, the augmenting graph approach has been successfully applied to solve MIS 
in various other special cases.  
However, our knowledge of augmenting graphs is still very limited, 
and we do not even know what the minimal {\it infinite} classes of augmenting graphs are.
In the present paper, we find an answer to this question and apply it to extend 
the area of polynomial-time solvability of the  maximum independent set problem. 
\end{abstract}


\section{Introduction}
In a graph, an {\it independent set} is a subset of pairwise non-adjacent vertices. 
For an input graph $G$, the {\sc maximum independent set} (MIS) problem asks to find the maximum cardinality (denoted $\alpha(G)$) of an independent set in $G$.
This is one of the central problems of combinatorial 
optimization with numerous applications and various connections to other problems in the area.

Like many important computational problems, {\sc maximum independent set} is NP-hard in general. However, for graphs with some special properties, 
the problem can be solved in polynomial time. This is the case, for instance, for the class of line graphs. The {\it line graph}
of a graph $G$ is the graph whose vertices represent the edges of $G$ with two vertices being adjacent if and only if the respective 
edges of $G$ share a vertex. Therefore, finding a maximum independent set in the line graph of $G$ is equivalent to finding 
a maximum matching in $G$, i.e. a maximum subset of edges no two of which share a vertex. 
The latter problem, unlike {\sc maximum independent set}, can be solved in polynomial time and 
the first polynomial-time algorithm to find a maximum matching in a graph was proposed by Edmonds \cite{Edmon1965} in 1965.
Lov\'asz and Plummer observed in their book ``Matching Theory''~\cite{LP86} that Edmonds' solution is 
``among the most involved of combinatorial algorithms.''  

In his solution to the maximum matching problem Edmonds  implemented the idea of augmenting chains proposed by Berge \cite{Berge1957}.
Later, in 1980, the same idea was used by Minty~\cite{Minty1980} and Sbihi \cite{Sbihi1980}, independently, in order to 
extend the solution of Edmonds from line graphs to claw-free graphs. After that,
for nearly two decades, the idea of augmenting chains did not see any further development and the result for claw-free 
graphs remained unimproved. 

In 1999, Alekseev \cite{Ale99} obtained a breakthrough result extending polynomial-time solvability of MIS from
claw-free to fork-free graphs. The crucial importance of this result is not only due to the fact that it extends 
the area of polynomial-time solvability of the problem. It also extends the {\it technique}. It shows that in addition
to augmenting chains there are other types of augmenting graphs and develops algorithms for detecting these graphs. 
In the same year, Mosca \cite{Mosca1999} discovered one more type of augmenting graphs (simple augmenting trees) 
and applied it to solve the problem in the class of $(P_6,C_4)$-free graphs. 
Since then it has been understood that the idea of augmenting chains is just a (very) special case of a general approach 
to solve the {\sc maximum independent set} problem, now known as the {\it augmenting graph technique}. 

In the last 15 years, the augmenting graph approach was frequently applied to various graph classes to design polynomial-time
algorithms for the {\sc maximum independent set} problem, and many new types of augmenting graphs have been discovered in 
the literature (see \cite{HL2005} for a survey). However, our knowledge in this area is still very limited. 
We do not even know what the minimal {\it infinite} classes of augmenting graphs are
(note that finding augmenting graphs from a finite collection is computationally a trivial task).
In the present paper, we answer this question. Our result allows us to identify new classes of graphs
with polynomial-time solvable {\sc maximum independent set} problem that 
extend some of the previously known results, such as algorithms for 
claw-free graphs and $(P_k,K_{1,t})$-free graphs.

The organization of the paper is as follows. In the rest of this section, we introduce basic terminology and notation. 
In Section~\ref{sec:aug}, we briefly review the idea of augmenting graphs. Then in Section~\ref{sec:aug-rams}
we present our Ramsey-type result about minimal {\it infinite} classes of augmenting graphs.
In Section~\ref{sec:app} we use this result to develop a polynomial-time solution in the class 
of $(S_{1,1,3},K_{p,p})$-free graphs that extends the class of claw-free graphs for any $p\ge 3$.
Finally, in Section~\ref{sec:con} we conclude the paper with a number of open problems.

\medskip
Given a graph $G$, we let $V(G)$ and $E(G)$ denote the vertex set and the edge set of $G$, respectively. 
For a vertex $v\in V(G)$, we let $N(v)$ denote the {\em neighbourhood} of $v$, i.e. the
set of vertices adjacent to $v$,
and for a set $U\subseteq V(G)$ we define $N(U) = \bigcup_{u \in U}N(u)$.
If $X\subseteq V(G)$, then $N_X(v)=N(v) \cap X$ is the neighbourhood of $v$ restricted to the set $X$, and similarly $N_X(U) = \bigcup_{u \in U}N_X(u)$.
The graph $G[X]$ is the subgraph of $G$ {\em induced} by $X$, i.e.
the graph obtained from $G$ by deleting every vertex not in $X$. 
As usual, $P_k$ denotes the chordless path on $k$ vertices and $K_{n,m}$ denotes 
the complete bipartite graph with parts of size~$n$ and~$m$. 
Also, $S_{i,j,k}$ denotes the tree with exactly three vertices of degree 1, being at distance $i,j,k$ from
the only vertex of degree 3. The graph $S_{1,1,1}=K_{1,3}$ is frequently referred to  as the {\it claw}
and $S_{1,1,2}$ as the {\it fork}. 

A class of graphs is said to be {\it hereditary} if for every graph $G$ in the class,
every induced subgraph of $G$ is also in the class. It is well known that a class of graphs is hereditary 
if and only if it can be characterized in terms of forbidden induced subgraphs. 
More precisely, for a set $M$ of graphs, let $Free(M)$ denote the class of graphs
containing no induced subgraphs from $M$. A class $X$ is hereditary if and only 
if $X=Free(M)$ for some set $M$. If $G\in Free(M)$, we say that $G$ is $M$-free.

A {\em bipartite} graph is a graph whose vertex set can be partitioned into two independent
sets. We denote such a graph by $(W,B,E)$, where $W$ and $B$ are 
the respective independent sets and $E$ is the set of edges.


\section{Augmenting Graphs}
\label{sec:aug}

Let $G$ be a graph,  $S$ be an independent set in $G$ and $R=V(G)\setminus S$.  
We say that the vertices in $S$ are \emph{white} and the vertices in $R$ are
\emph{black}.  Consider two subsets $W \subseteq S$ and $B \subseteq R$. 
Note that $W$ is an independent set. If $B$ also is an independent set, $|B|>|W|$ and $N(B)\cap S
\subseteq W$, we say that the bipartite graph $H=G[W\cup B]$ is \emph{augmenting for the set} $S$. 

Clearly, if $G$ contains an augmenting graph $H=G[W\cup B]$ for $S$, then $S$ is not maximum, because 
$T:=(S \setminus W)\cup B$ is an independent set larger than $S$, in which case we say that 
$T$ is obtained from $S$ by {\em $H$-augmentation}. 
On the other hand, if $S$ is not maximum and $T$ is a larger independent set, 
then the bipartite subgraph of $G$ induced by $(T\setminus S)$ and $(S \setminus T)$
is augmenting for $S$. Thus we obtain the following well-known result.

\begin{theorem}[Augmenting Graph Theorem]
An independent set $S$ in a graph $G$ is maximum if and only if there are no
augmenting graphs for $S$.
\end{theorem}

This theorem suggests the following general approach to find a maximum
independent set in a graph $G$: begin with any independent set $S$ in $G$ and
as long as $S$ admits an augmenting graph $H$, apply $H$-augmentations to $S$.
Clearly the problem of finding augmenting graphs is NP-hard in general, as the
maximum independent set problem is NP-hard. However, for graphs in some special 
classes this approach can lead to polynomial-time algorithms, which is the case 
for line graphs (the maximum matching problem), claw-free graphs ~\cite{Minty1980,Sbihi1980},
fork-free graphs \cite{Ale99} and many other classes (see \cite{HL2005} for a survey). 

To effectively apply this approach to a particular class of graphs, 
we first have to characterize the augmenting graphs in the class 
and then develop polynomial-time algorithms for detecting these graphs.

Obviously, if the list of augmenting graphs is finite, then all of them
can be detected in polynomial time. Therefore, only {\it infinite} families of
augmenting graphs are of interest. In Section~\ref{sec:aug-rams}, we show that,
with the restriction to hereditary classes, there are exactly three minimal
infinite families of augmenting graphs.

\section{Minimal infinite classes of augmenting graphs}
\label{sec:aug-rams}
According to Ramsey's theorem, every graph with sufficiently many vertices contains either a ``large'' independent set 
or a ``large'' clique. This result can also be interpreted as follows:
in the family of hereditary classes there are precisely two minimal infinite classes of 
graphs, the class of edgeless graphs and the class of complete graphs. Indeed, 
each of these two classes is infinite and any hereditary class excluding at least 
one edgeless graph and one complete graph is finite (since the number of vertices 
in graphs in this class is bounded by a Ramsey number).  
In the present section, we prove a result of the same flavour. 
To formally state the result, we need to update some terminology related to augmenting graphs.

\medskip

If $H$ is an augmenting graph  for an independent set $S$, then it may happen that
a proper induced subgraph of $H$ is also augmenting for the same set. For instance,   
if a star~$K_{1,p}$ with $p>2$ is augmenting for $S$, then any induced $K_{1,2}$ of this 
star is also augmenting for $S$. This observation motivates the notion of a {\it minimal
augmenting graph for} $S$, i.e. an augmenting graph containing no proper induced subgraph 
which is also augmenting for $S$. In \cite{LM2008}, it was proved that an augmenting graph $H=(W,B,E)$ is 
minimal for an independent set $S$ if and only if it possesses each of the following three properties:
\begin{enumerate}[(a)]
	\item \label{prop:irred-a} $|W| = |B| - 1$;
	\item \label{prop:irred-b} for every nonempty subset $A \subseteq W, |A| < |N(A)\cap B|$;
	\item \label{prop:irred-c} $H$ is connected.
\end{enumerate}
In what follows, we will call any bipartite graph $H=(W,B,E)$  satisfying Properties (\ref{prop:irred-a}), (\ref{prop:irred-b}) and~(\ref{prop:irred-c}) an {\it irreducible} graph,
without any reference to a specific independent set.  
Clearly, if an independent set $S$ admits an augmenting graph, then it also admits an augmenting graph which is irreducible.
Therefore, the universe of augmenting graphs can be restricted, without loss of generality, to irreducible ones. 
 
For an arbitrary set $\cal C$ of graphs, let ${\cal C}^i$ denote the set of irreducible graphs in $\cal C$, in which case 
we say that the set ${\cal C}^i$ is {\it generated}  by $\cal C$.
Our goal is to identify minimal infinite sets of irreducible graphs generated by {\it hereditary} classes. One such set is
\begin{itemize}
\item the set $\cal P$ of chordless paths of even length.
\end{itemize}
Clearly, each graph in this set 
is irreducible. Moreover, $\cal P$ coincides with the set of irreducible graphs in the class of claw-free graphs.
Indeed, by definition, every irreducible graph is bipartite, and any bipartite claw-free graph has maximum vertex degree at most 2
(otherwise a claw arises). In other words, a connected bipartite claw-free graph is either a path or an even cycle.
Neither paths of odd length nor even cycles are augmenting (as they have equally many black and white 
vertices). Therefore, the set of irreducible claw-free graphs coincides with $\cal P$.
Clearly, the set $\cal P$ is infinite. Moreover, it is a {\it minimal} infinite class generated by a hereditary class.
Indeed, let $X=Free(M)$ be a hereditary class defined by a set  $M$ of forbidden induced subgraphs. 
If $M$ does not contain any graph every connected component of which is a path, then $X$ contains all graphs from $\cal P$. 
Otherwise,~$X$ contains only finitely many graphs from $\cal P$.

Similarly, it is easy to check that
\begin{itemize}
\item the set $\cal K$ of complete bipartite graphs $K_{k,k+1}$ and
\item the set $\cal T$ of simple trees $T_k$, i.e. graphs formed from a star $K_{1,k}$
by subdividing each edge exactly once (see Figure~\ref{fig:agc} for an example)
\end{itemize}
are minimal infinite sets of irreducible graphs generated by hereditary classes. Below we show that $\cal P, K, T$ are 
the {\it only} sets of irreducible graphs with this property.
To prove our result, we need the following lemma.

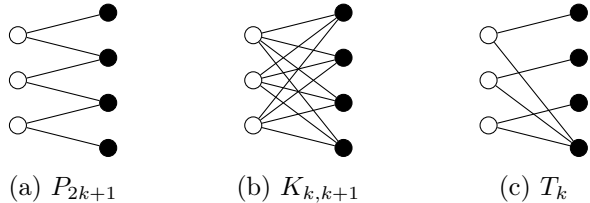
\begin{figure}[ht!]
\begin{center}
\begin{subfigure}{0.2\textwidth}
\centering
\begin{tikzpicture}[scale=.6,auto=left]
\node[w_vertex] (b) at (0,0.5) { };
\node[w_vertex] (d) at (0,1.5) { };
\node[w_vertex] (f) at (0,2.5) { };
\node[b_vertex] (a) at (2,0) { };
\node[b_vertex] (c) at (2,1) { };
\node[b_vertex] (e) at (2,2) { };
\node[b_vertex] (g) at (2,3) { };
\foreach \from/\to in {a/b,b/c,c/d,d/e,e/f,f/g}
 \draw (\from) -- (\to);
\end{tikzpicture}
\subcaption{$P_{2k+1}$}
\end{subfigure}
\begin{subfigure}{0.2\textwidth}
\centering
\begin{tikzpicture}[scale=.6,auto=left]
\node[w_vertex] (b) at (0,0.5) { };
\node[w_vertex] (d) at (0,1.5) { };
\node[w_vertex] (f) at (0,2.5) { };
\node[b_vertex] (a) at (2,0) { };
\node[b_vertex] (c) at (2,1) { };
\node[b_vertex] (e) at (2,2) { };
\node[b_vertex] (g) at (2,3) { };
\foreach \from in {a,c,e,g}
 \foreach \to in {b,d,f}
  \draw (\from) -- (\to);
\end{tikzpicture}
\subcaption{$K_{k,k+1}$}
\end{subfigure}
\begin{subfigure}{0.2\textwidth}
\centering
\begin{tikzpicture}[scale=.6,auto=left]
\node[w_vertex] (b) at (0,0.5) { };
\node[w_vertex] (d) at (0,1.5) { };
\node[w_vertex] (f) at (0,2.5) { };
\node[b_vertex] (a) at (2,0) { };
\node[b_vertex] (c) at (2,1) { };
\node[b_vertex] (e) at (2,2) { };
\node[b_vertex] (g) at (2,3) { };
\foreach \from/\to in {a/b,a/d,a/f,b/c,d/e,f/g}
 \draw (\from) -- (\to);
\end{tikzpicture}
\subcaption{$T_k$}
\end{subfigure}
\caption{The three special families of augmenting graphs.}\label{fig:agc}
\end{center}
\end{figure}

\begin{lemma}[\cite{DDL2013}]\label{lem:Ramsey}
For any natural numbers $t$ and $p$, there is a number $N(t,p)$ such that every
bipartite graph with a matching of size at least $N(t,p)$ contains either a
bi-clique $K_{t,t}$ or an induced matching on $p$ edges.
\end{lemma}

\begin{theorem}\label{th:MinClasses}
Let ${\cal C}$ be a hereditary class of graphs and 
let ${\cal C}^i$ be the set of irreducible graphs generated by ${\cal C}$. 
If ${\cal C}^i$ is infinite, then it contains at least one of $\cal P$, $\cal K$ or $\cal T$.
\end{theorem}

\begin{proof}
Suppose the theorem is false, i.e. ${\cal C}^i$ is infinite, but there is a $t$ such that 
${\cal C}^i$ does not contain any $P_t$, $K_{t-1,t}$ or $T_t$.  The graphs in 
${\cal C}^i$ are connected, but are $P_t$-free, so there must be graphs in
${\cal C}^i$ with vertices of arbitrarily large degree, in particular, of degree at least $N(t,t)+2$.

Consider a graph $G = (W,B,E)$ in ${\cal C}^i$. 
By Property~(\ref{prop:irred-b}) of irreducible graphs, for any subset $W'$ of $W$, we must have $|W'| \leq |N_B(W') \cap B|$ and
therefore, by Hall's Marriage Theorem, there must be
a matching $M$ from $W$ to $B$ (one vertex of $B$ remain unmatched to any
vertex of $W$ since $|B|=|W|+1$).

Now let $G = (W,B,E)$ be any graph in ${\cal C}^i$ containing a vertex $x$ of
degree at least $N(t,t)+2$.  Let $X$ be the
set of vertices in the neighbourhood of $x$ which form part of the matching
$M$, but are not matched with $x$. $X$ must contain at least $N(t,t)$ vertices.
Let~$Y$ be the set of vertices which $M$ matches to the vertices of $X$. Then
$G[X\cup Y]$ contains a matching of size $N(t,t)$, but it is $K_{t-1,t}$-free and therefore $K_{t,t}$-free. This
implies, by Lemma~\ref{lem:Ramsey}, that it must contain an induced matching on $t$ edges. Let $Z$
be the set of vertices that occur in this induced matching. Then
$G[Z\cup\{x\}]$ forms a $T_t$, so $T_t \in {\cal C}$ and therefore $T_t \in
{\cal C}^i$. This contradiction completes the proof.
\end{proof}

This theorem implies that for any $t$ the class of $(P_t,K_{t,t},T_t)$-free graphs 
contains only finitely many irreducible graphs. Therefore:

\begin{corollary}\label{cor:main}
For positive integers $i,j,k$, the {\sc maximum independent set} problem can be
solved in the class of $(P_i,K_{j,j},T_k)$-free graphs in polynomial time.
\end{corollary}

This result generalizes the polynomial-time 
solvability of the problem  in the class of $(P_k,K_{1,t})$-free graphs 
proved in \cite{LR2003}. Also, it was recently shown in \cite{new} that 
the problem can be solved in polynomial time in a subclass of $(P_i,K_{j,j},T_k)$-free graphs
defined by two additional forbidden induced subgraphs. Corollary~\ref{cor:main} also generalizes 
this result.  

\section{Independent sets in $(S_{1,1,3},K_{p,p})$-free graphs}
\label{sec:app}

In this section, we solve the {\sc maximum independent set} problem
in polynomial time for $(S_{1,1,3},K_{p,p})$-free graphs. Observe that for $p>2$ this class contains 
all claw-free graphs. Therefore, our result generalizes the solution for claw-free 
graphs and hence the solution of the {\sc maximum matching} problem. 

We first describe the structure of irreducible graphs in our class (Section~\ref{sec:str})
and then show how to find these graphs in polynomial time (Section~\ref{sec:alg}).

\subsection{The structure of augmenting $(S_{1,1,3},K_{p,p})$-free graphs}
\label{sec:str}

According to Theorem~\ref{th:MinClasses}, there are only finitely many $(S_{1,1,3},K_{p,p})$-free graphs that are irreducible and contain neither long induced paths nor large induced simple trees.
Therefore, in this section we restrict ourselves to describing the irreducible graphs
containing either a long induced path (Lemma~\ref{lem:AugGraphWithLongPath}) or a large induced simple tree (Lemma~\ref{lem:MinAugGraphWithAp+2}).
We start with the structure of $S_{1,1,3}$-free bipartite graphs containing a long induced path.

\begin{lemma}\label{lem:AugGraphWithLongPath}
	Let $H=(W,B,E)$ be a connected $S_{1,1,3}$-free bipartite graph
	containing a $P_8$ as an induced subgraph. Then $H$ is either 
	a chordless path or a chordless cycle.
\end{lemma}
\begin{proof}
Assume that $H$ is not a chordless path. We will show that $H$ is a chordless cycle. 
	
	Let $P$ be an induced path of maximum length in $H$. Since $P$ has
	at least eight vertices, each of the parts of $H$ contains at least three
	\textit{internal vertices} of $P$, i.e. vertices different from the endpoints of $P$. Let $x$ denote a vertex of $H$ outside
	$P$, which has a neighbour in $P$. Assume without loss of generality
	that $x$ belongs to $W$. We claim that $x$ has no neighbours
	among the internal vertices of $P$.
	Suppose for a contradiction that $x$ is adjacent to an internal vertex of $P$.
	Then at least one of the following two cases takes place:
	\begin{enumerate}
		\item There are two consecutive internal vertices
		$b_1,b_2 \in B$ of $P$ 	such that $(x,b_1) \notin E$ and $(x,b_2) \in E$. Then $x,b_1,b_2$ and
		the three vertices of $P$ adjacent to $b_1$ or to $b_2$ induce
		an $S_{1,1,3}$.
		
		
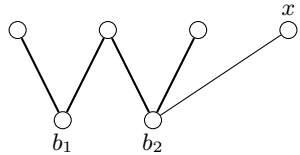
\begin{figure}[ht!]
	      		\centering
\begin{tikzpicture}
	  		[scale=.6,auto=left]
			\node[w_vertex] (w1) at (-1,2) { }; 	  		
	  		
			\node[w_vertex] (b1) at (0,0) { }; 
			\coordinate [label=center:\footnotesize{$b_1$}] (b_1) at (0,-0.5);
			
			\node[w_vertex] (w2) at (1,2) { }; 	 
			
			\node[w_vertex] (b2) at (2,0) { };
			\coordinate [label=center:\footnotesize{$b_2$}] (b_2) at (2,-0.5);
			
			\node[w_vertex] (w3) at (3,2) { }; 	 
			
			\node[w_vertex] (x) at (5,2) { };
			\coordinate [label=center:\footnotesize{$x$}] (x_) at (5,2.45);

			\foreach \from/\to in {w1/b1,b1/w2,w2/b2,b2/w3}
	    		\draw[thick] (\from) -- (\to);
	    		
	    		\foreach \from/\to in {x/b2}
	    		\draw (\from) -- (\to);		
\end{tikzpicture}
			\label{two_black_fig}
	       \end{figure}
		
		\item There are three consecutive internal vertices
		$b_1,b_2,b_3 \in B$
		of $P$ such that $x$ is adjacent to all of them. Then $x,b_1,b_3$,
		the two neighbours of $b_1$ in $P$ and any neighbour of $b_3$ in $P$
		induce an $S_{1,1,3}$.
		
		\begin{figure}[ht!]
	      		\centering
\begin{tikzpicture}
	  		[scale=.6,auto=left]
			\node[w_vertex] (w1) at (-1,2) { }; 	  		
	  		
			\node[w_vertex] (b1) at (0,0) { }; 
			\coordinate [label=center:\footnotesize{$b_1$}] (b_1) at (0,-0.5);
			
			\node[w_vertex] (w2) at (1,2) { }; 	 
			
			\node[w_vertex] (b2) at (2,0) { };
			\coordinate [label=center:\footnotesize{$b_2$}] (b_2) at (2,-0.5);
			
			\node[w_vertex] (w3) at (3,2) { }; 	 
			
			\node[w_vertex] (b3) at (4,0) { };
			\coordinate [label=center:\footnotesize{$b_3$}] (b_3) at (4,-0.5);
			
			\node[w_vertex] (w4) at (5,2) { }; 	 
			
			\node[w_vertex] (x) at (7,2) { };
			\coordinate [label=center:\footnotesize{$x$}] (x_) at (7,2.45);

			\foreach \from/\to in {w1/b1,b1/w2,w2/b2,b2/w3,w3/b3,b3/w4}
	    		\draw[thick] (\from) -- (\to);
	    		
	    		\foreach \from/\to in {x/b1,x/b2,x/b3}
	    		\draw (\from) -- (\to);		
		\end{tikzpicture}
			\label{three_black_fig}
	       \end{figure}
	\end{enumerate}
	
	This contradiction and the maximality of $P$ imply that $x$ has exactly two
	neighbours in~$P$, namely the first and the last vertex of the path.
	Finally, the graph $H$ does not contain any other vertices, since otherwise $H$
	would contain a vertex $y$ outside of $P$ distinct from~$x$ such that:
	\begin{itemize}
		\item either $y$ is not adjacent to $x$ and has exactly two neighbours 
		in $P$, which are the end-vertices of the path,
		\item or $y$ is adjacent to $x$ and has no neighbours in $P$.
	\end{itemize}
	It is easy to see that in both cases an induced $S_{1,1,3}$ would arise.
\end{proof}

\medskip
Next, we describe the structure of $S_{1,1,3}$-free bipartite graphs containing a large induced simple tree, i.e. a graph of the form $T_k$ 
(see Figure~\ref{fig:agc}). Suppose that a bipartite $S_{1,1,3}$-free graph contains an induced copy of $T_k$ with $k\ge 3$
and let $T$ be such a copy which is maximal with respect to inclusion. We define the following
\begin{itemize}
\item $u$ the central vertex of $T$, 
\item $A_0 = \{a_1, \ldots, a_k\}$ the set of neighbours of $u$ in $T$, 
\item $B_0 = \{b_1, \ldots, b_k\}$ the set of leaves of $T$ with $a_ib_i \in E$ for $i=1,\ldots,k$,
\item $B_1 = N(B_0) \setminus A_0$,
\item $B_1' \subseteq B_1$ the set of vertices not in $A_0$ with exactly one neighbour in $B_0$, 
\item $B_1'' \subseteq B_1$ the set of vertices which are adjacent to all the vertices of $B_0$, 
\item $A_1=N(A_0) \setminus (\{u\} \cup B_0)$,
\item $C = N(u) \setminus (A_0 \cup B_1)$,
\item $D_1=N(A_1) \setminus (C \cup A_0 \cup B_1)$,
\item $D_2=N(B_1) \setminus (\{u\} \cup B_0 \cup A_1)$.
\end{itemize}
	
	\begin{figure}[ht!]
		\centering
\begin{tikzpicture}
	  		[scale=.6,auto=left]
			\node[w_vertex] (u) at (0,0) { }; 
			\coordinate [label=center:\footnotesize{$u$}] (u_) at (0,-0.45);
			
			\draw (0,4) ellipse [x radius = 2, y radius = 1];
			\coordinate [label=center:$C$] (C) at (0,5.5);
			
			\draw (u) -- (-1.96,3.8);	
			\draw (u) -- (1.96,3.8);	

			\node[w_vertex] (b1) at (3,4) { };
			\coordinate [label=center:\footnotesize{$a_1$}] (a_1) at (3.1,4.4);
			\node[w_vertex] (b2) at (4,4) { };
			\coordinate [label=center:\footnotesize{$a_2$}] (a_2) at (4.1,4.4);
			\node[point] (px1) at (4.66, 4) {};
			\node[point] (px2) at (5, 4) {};
                     \node[point] (px3) at (5.33, 4) {};
			\node[w_vertex] (bs) at (6,4) { };
			\coordinate [label=center:\footnotesize{$a_k$}] (a_s) at (6.1,4.4);
			\draw [gray, decorate, decoration={brace, amplitude=8}] (2.7,4.7)  --  (6.3,4.7);
			\coordinate [label=center:$A_0$] (A0) at (4.5,5.5);
			\node[w_vertex] (a1) at (3,0)   { };
			\coordinate [label=center:\footnotesize{$b_1$}] (b_1) at (3.1,-0.45);
			\node[w_vertex] (a2) at (4,0)   { };
			\coordinate [label=center:\footnotesize{$b_2$}] (b_2) at (4.1,-0.45);
			\node[point] (py1) at (4.66, 0) {};
			\node[point] (py2) at (5, 0) {};
                                \node[point] (py3) at (5.33, 0) {};
			\node[w_vertex] (as) at (6,0)   { };
			\coordinate [label=center:\footnotesize{$b_k$}] (b_s) at (6.1,-0.45);
			\draw [gray, decorate, decoration={brace, amplitude=8}] (6.3,-0.7)  --  (2.7,-0.7);
			\coordinate [label=center:$B_0$] (B0) at (4.5,-1.5);

			\foreach \from/\to in {b1/a1,b2/a2,bs/as}
	    		\draw (\from) -- (\to);
	    		
	    		\foreach \from/\to in {u/b1,u/b2,u/bs}
	    		\draw (\from) -- (\to);

			\draw (9,4) ellipse [x radius = 2.1, y radius = 1.1]; 
			\coordinate [label=center:$B_1$] (B1) at (9,5.5);
			
				\draw (8,4) ellipse [x radius = 0.95, y radius = 0.5]; 
				\coordinate [label=center:$B_1'$] (B1') at (8,4);
				
				\draw (10,4) ellipse [x radius = 0.95, y radius = 0.5]; 
				\coordinate [label=center:$B_1''$] (B1'') at (10,4);
			
			\draw (9,0) ellipse [x radius = 2, y radius = 1];
			\coordinate [label=center:$A_1$] (A1) at (9,-1.5);				

			\draw (14,0) ellipse [x radius = 2, y radius = 1]; 
			\coordinate [label=center:$D_2$] (D2) at (14,-1.5);
			
			\draw (14,4) ellipse [x radius = 2, y radius = 1]; 
			\coordinate [label=center:$D_1$] (D1) at (14,5.5);
			
		\end{tikzpicture}
 		\caption{A bipartite graph containing an induced copy of $T_k$ centred at vertex $u$.}
 		\label{Ap}
	\end{figure}
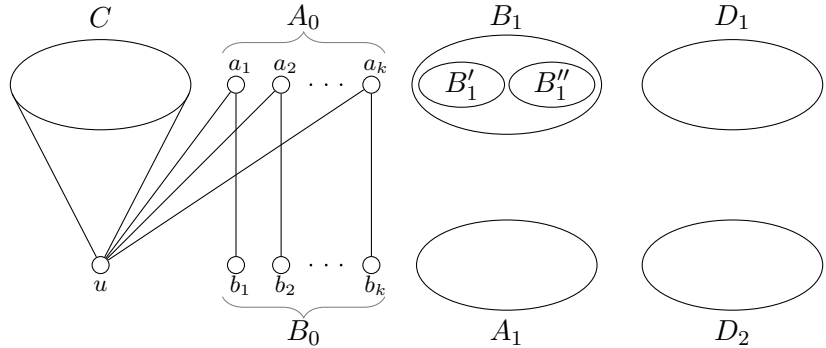

\begin{lemma}\label{lem:Ap}
If $T$ is an inclusionwise maximal induced copy of $T_k$ with $k\ge 3$ in a bipartite $S_{1,1,3}$-free graph, then the following statements hold:
\begin{enumerate}[(i)]
\item \label{stat:struct1} Every vertex of $B_1$ is adjacent to $u$.
\item \label{stat:struct2} Every vertex of $A_1$ is adjacent to every vertex of $A_0$. 
\item \label{stat:struct3} No vertex of $C$ has a neighbour outside of $\{u\} \cup A_1$.
\item \label{stat:struct4} No vertex of $D_1$ has a neighbour outside of $A_1$.
\item \label{stat:struct5} $B_1=B_1'\cup B_1''$. 
\item \label{stat:struct6} $B_1'=\emptyset$ or $B_1''=\emptyset$.
\item \label{stat:struct7} No vertex of $B_1'$ has a neighbour outside of $\{u\} \cup B_0  \cup A_1$.
\item \label{stat:struct8} No vertex of $D_2$ has a neighbour outside of $B_1$.
\end{enumerate}
\end{lemma}

\begin{proof}
To prove Statement~(\ref{stat:struct1}), suppose a vertex $y\in B_1$ is non-adjacent to $u$. 
Let $b_k$ be a neighbour of $y$ in $B_0$, and $a_i, a_j, a_k$ be three distinct vertices from $A_0$.
Then the vertex set $\{ u, a_i, a_j, a_k, b_k, y \}$ induces an $S_{1,1,3}$. 

To prove Statement~(\ref{stat:struct2}), suppose that $x \in A_1$ has a non-neighbour in
$A_0$, say $a_i$. By definition, $x$ has a neighbour in $A_0$, say $a_j$.
But then the vertex set $\{ a_j, b_j, x, u, a_i, b_i \}$ induces an $S_{1,1,3}$. 

To prove Statement~(\ref{stat:struct3}), suppose a vertex $x \in C$  has a neighbour 
$y \notin \{u\} \cup A_1$. Then the set
$\{u,x,y\} \cup A_0 \cup B_0$ induces a copy of a simple tree which properly contains $T$,
contradicting the maximality of $T$. 

To prove Statement~(\ref{stat:struct4}), suppose a vertex $y \in D_1$ has a neighbour 
$z \notin A_1$ and let $x$ be a neighbour of $y$ in $A_1$.
Observe that $y$ is not adjacent to $u$, since otherwise $y$ would belong to $C$.
But then by Statement~(\ref{stat:struct2}), the vertex set $\{ a_1, b_1, u, x, y, z \}$ induces an $S_{1,1,3}$.

To prove Statement~(\ref{stat:struct5}), suppose a vertex $x$ in $B_1$  has at least 
two neighbours, say~$b_i$ and~$b_j$, and at least one non-neighbour, say $b_k$,
in $B_0$. Then by Statement~(\ref{stat:struct1}), the vertex set $\{ x, b_i, b_j, u, a_k, b_k \}$ induces an~$S_{1,1,3}$.

To prove Statement~(\ref{stat:struct6}), suppose that each of $B_1'$ and $B_1''$ contains at least
one vertex, say $x \in B_1'$ and $y \in B_1''$. 
Then the vertex set $\{ b_i, a_i, x, y, b_j, a_j \}$ induces an $S_{1,1,3}$, 
where $b_i$ is the neighbour of $x$ in $B_0$ and
$b_j$ is any vertex of $B_0$ different from $b_i$.

To prove Statement~(\ref{stat:struct7}),  suppose a vertex $x \in B_1'$  has a neighbour
$y \notin \{u\} \cup B_0  \cup A_1$. Then by Statement~(\ref{stat:struct1}), the
vertex set $\{ x, b_i, y, u, a_j, b_j\}$ induces an $S_{1,1,3}$, 
where $b_i \in B_0$ is a neighbour and $b_j \in B_0$ is a non-neighbour 
of $x$, respectively.
		
To prove Statement~(\ref{stat:struct8}), suppose a vertex $y \in D_2$  has a neighbour 
$z \notin B_1$. Then by Statements~(\ref{stat:struct1}) and~(\ref{stat:struct3}), the vertex set $\{ u, a_1, a_2, x, y, z \}$
induces an $S_{1,1,3}$, where $x \in B_1$ is a neighbour of $y$.
\end{proof}
Note that if the graph in the above lemma is connected, it follows that every vertex of the graph belongs to $\{u\} \cup A_0 \cup A_1 \cup B_0 \cup B_1 \cup C \cup D_0 \cup D_1$.

\medskip
From now on, we deal with bipartite graphs that are irreducible,
i.e. we assume that their vertices are coloured black and white and that they satisfy Properties (\ref{prop:irred-a}), (\ref{prop:irred-b}) and~(\ref{prop:irred-c}) of irreducible graphs.
Our goal is to prove that if $H=(W,B,E)$ is an irreducible $(S_{1,1,3},K_{p,p})$-free graph containing an induced copy of $T_k$ with $k\ge p+2$, 
then $H$ differs from a simple tree only by finitely many vertices. 
To prove this result, we first show in the next lemma that we can always assume that an induced copy of 
$T_k$ with $k\ge p+2$ appears in $H$ with its central vertex being black.

\begin{lemma}\label{Ap+2_subgraph}
Let $p \in \mathbb{N}$ and $H=(W,B,E)$ be an irreducible $(S_{1,1,3},K_{p,p})$-free graph.
If $H$ contains an induced copy of the graph $T_{p+2}$, then it contains an induced copy of 
$T_{p+2}$ in which the central vertex is black. 
\end{lemma}
\begin{proof}
Let $T$ be an inclusionwise maximal induced copy of $T_k$ with $k\ge p+2$.
If the centre $u$ of $T$ is black, then we are done. So assume that $u$, as well as the centre of any other induced $T_{p+2}$, is white.
This assumption implies, by the definition of $B_1''$, that $B_1''$ is empty, since otherwise any vertex of $B_1''$ is a black centre of an induced $T_{p+2}$. 
Now from Statements (\ref{stat:struct3}), (\ref{stat:struct4}), (\ref{stat:struct5}) and~(\ref{stat:struct7}) of Lemma~\ref{lem:Ap} we conclude that  
$W = \{u\} \cup B_0 \cup A_1$ and $B = C \cup A_0 \cup B_1' \cup D_1$. Therefore, by Property~(\ref{prop:irred-a}) of irreducible graphs 
	\begin{equation}\label{eq:A1_cardinality}
		|A_1| = |C| + |B_1'| + |D_1| - 2 \geq |B_1'| - 2.
	\end{equation}
From the definition of $B_1'$ and Property~(\ref{prop:irred-b}) of irreducible graphs
	it follows that $|B_1'| \geq p+2$. This together with Equation~(\ref{eq:A1_cardinality})
	and Statement~(\ref{stat:struct2}) of Lemma~\ref{lem:Ap} implies that $A_0 \cup A_1$ induces a subgraph containing $K_{p,p}$.
	This is a contradiction. Therefore $H$ contains an induced copy of $T_{p+2}$ in which the central vertex is black. 
\end{proof}

We will now show that the structure of every irreducible $(S_{1,1,3},K_{p,p})$-free graph containing a large induced copy of $T_{k}$
is very close to the structure of a simple tree. More formally, we will say that a graph $H$ is an $s$-{\it extension} of 
a simple tree if it can be reduced to a simple tree by deleting at most $s$ vertices.

\begin{lemma}\label{lem:MinAugGraphWithAp+2}
	Let $p \in \mathbb{N}$ and $H=(W,B,E)$ be an irreducible $(S_{1,1,3},K_{p,p})$-free graph containing $T_{p+2}$ as an 
	induced subgraph. Then $H$ is a $4p$-extension of a simple tree $T_k$ with $k\ge p+2$ and in which the central vertex is black. 
\end{lemma}
\begin{proof}
As before, let $T$ denote an inclusionwise maximal induced copy of $T_k$ with $k\ge p+2$ and assume by Lemma~\ref{Ap+2_subgraph} 
that the centre $u$ of $T$ is black. The fact that $H$ is $K_{p,p}$-free together with Statement~(\ref{stat:struct2}) of Lemma~\ref{lem:Ap} implies
	that
	\begin{equation}\label{A1_cardinality}
		|A_1| < p.
	\end{equation}
Similarly, the fact that $H$ is $K_{p,p}$-free together with the definition of $B_1''$  implies
	that
 	\begin{equation}\label{B1''_cardinality}
 		|B_1''| < p.
 	\end{equation}
Statements~(\ref{stat:struct3}) and~(\ref{stat:struct4}) of Lemma~\ref{lem:Ap} together with Property~(\ref{prop:irred-b}) of irreducible graphs and inequality~(\ref{A1_cardinality}) imply that 
	\begin{equation}\label{D1_cardinality}
		|C| + |D_1| < |A_1| + 1 \leq p.
	\end{equation}
Statements (\ref{stat:struct3}), (\ref{stat:struct4}), (\ref{stat:struct5}), (\ref{stat:struct7}) and~(\ref{stat:struct8}) of Lemma~\ref{lem:Ap} imply that 
	$W = C \cup A_0 \cup B_1' \cup B_1'' \cup D_1$
	and $B = \{u\} \cup B_0 \cup A_1 \cup D_2$. Therefore, by Property~(\ref{prop:irred-a})
	of irreducible graphs 
	$$
		|C| + |B_1'| + |B_1''| + |D_1| = |A_1| + |D_2|.
	$$
According to Statement~(\ref{stat:struct6}) of Lemma~\ref{lem:Ap} there are two cases:
	\begin{enumerate}
		\item $B_1' = \emptyset$. This together with inequalities 
		(\ref{B1''_cardinality}) and (\ref{D1_cardinality})
		implies that
		$$
			|A_1| + |D_2| = |C| + |B_1''| + |D_1| < 2p,
		$$
		i.e. the graph $H$ contains less than $4p$ vertices besides the $2k+1$ vertices of $T_k$.
	
		\item $B_1'' = \emptyset$. In this case, 
		Statement~(\ref{stat:struct7}) of Lemma~\ref{lem:Ap} implies that $D_2$ is also empty 
		and taking into account inequality (\ref{A1_cardinality}) we have 
		$$
			|C| + |B_1'| + |D_1| = |A_1| < p,
		$$
		i.e. the graph $H$ contains less than $2p$ vertices besides the $2k+1$ vertices of $T_k$.
	\end{enumerate}
\end{proof}

\begin{theorem}\label{th:MinAugStruct}
Let $p \in \mathbb{N}$ and let $H=(W,B,E)$ be an irreducible $(S_{1,1,3},K_{p,p})$-free graph.
Then $H$ is either
\begin{itemize}
\item an induced path of even length or
\item a $4p$-extension of a simple tree $T_k$ with $k\ge p+2$ or	
\item a member of the finite set of $(P_8,T_{p+2},K_{p,p})$-free irreducible graphs.
\end{itemize}
\end{theorem}
\begin{proof}
If $H$ contains an induced $P_8$, then by Lemma~\ref{lem:AugGraphWithLongPath} the graph $H$ is an induced path of even length.
If $H$ contains an induced copy of $T_{p+2}$, then by Lemma~\ref{lem:MinAugGraphWithAp+2} the graph $H$ is a $4p$-extension of a simple tree $T_k$ with $k\ge p+2$. 
If $H$ contains neither $P_8$ nor $T_{p+2}$, then it belongs to a finite collection of irreducible graphs by Theorem~\ref{th:MinClasses}.
\end{proof}

\subsection{Finding augmenting $(S_{1,1,3},K_{p,p})$-free graphs}
\label{sec:alg}
In this section we deal with the problem of finding augmenting graphs in $(S_{1,1,3},K_{p,p})$-free graphs.
According to Theorem~\ref{th:MinAugStruct}, this problem consists of two main subproblems:
finding augmenting paths and finding extensions of simple trees. 
The first of these was solved in \cite{GHL2006} even for more general graphs, namely for $S_{1,2,3}$-free graphs. 
In Lemma~\ref{lem:FindExtOfSimpleTree} we solve the second subproblem. Then in Theorem~\ref{thm:FindMinAugGraph}
we summarize our arguments and present a polynomial-time solution to the {\sc maximum independent set}
problem in the class of $(S_{1,1,3},K_{p,p})$-free graphs.

\begin{lemma}
\label{lem:FindExtOfSimpleTree}
	Let $p \geq 2$, $G=(V,E)$ be an $(S_{1,1,3}, K_{p,p})$-free graph and $S \subseteq V$
	be an independent set in $G$. 
	Then in polynomial time one can determine whether $G$ contains an irreducible augmenting graph for $S$ 
	which is a $4p$-extension of a simple tree $T_k$ with $k\ge p+2$.
\end{lemma}
\begin{proof}
Suppose that $G$ contains an irreducible augmenting graph $H=(W,B,E')$ for $S$ which is
a $4p$-extension of a simple tree $T_k$ with $k\ge p+2$. As before, we denote the centre 
of $T_k$ by $u$ and by Lemma~\ref{lem:MinAugGraphWithAp+2} we may assume it is black. Also, let $A_0$ and $B_0$ denote the sets of white and black non-centre vertices 
of $T_k$, respectively. Finally, let $Q_1$ denote the set of additional white vertices of $H$
and let $Q_2$ denote the set of additional black vertices of $H$. Since $H$ is irreducible, it follows that $|Q_1|=|Q_2|$.

	\begin{figure}[ht!]
		\centering
\begin{tikzpicture}
	  		[scale=.6,auto=left]
			\node[b_vertex] (u) at (0,0) { }; 
			\coordinate [label=center:\footnotesize{$u$}] (u_) at (0,-0.45);

			\node[w_vertex] (b1) at (3,4) { };
			\coordinate [label=center:\footnotesize{$a_1$}] (a_1) at (3.1,4.4);
			\node[w_vertex] (b2) at (4,4) { };
			\coordinate [label=center:\footnotesize{$a_2$}] (a_2) at (4.1,4.4);
			\node[point] (px1) at (4.66, 4) {};
			\node[point] (px2) at (5, 4) {};
                     \node[point] (px3) at (5.33, 4) {};
			\node[w_vertex] (bs) at (6,4) { };
			\coordinate [label=center:\footnotesize{$a_s$}] (a_s) at (6.1,4.4);
			\draw [gray, decorate, decoration={brace, amplitude=8}] (2.7,4.7)  --  (6.3,4.7);
			\coordinate [label=center:$A_0$] (A0) at (4.5,5.5);
			\node[b_vertex] (a1) at (3,0)   { };
			\coordinate [label=center:\footnotesize{$b_1$}] (b_1) at (3.1,-0.45);
			\node[b_vertex] (a2) at (4,0)   { };
			\coordinate [label=center:\footnotesize{$b_2$}] (b_2) at (4.1,-0.45);
			\node[point] (py1) at (4.66, 0) {};
			\node[point] (py2) at (5, 0) {};
                                \node[point] (py3) at (5.33, 0) {};
			\node[b_vertex] (as) at (6,0)   { };
			\coordinate [label=center:\footnotesize{$b_s$}] (b_s) at (6.1,-0.45);
			\draw [gray, decorate, decoration={brace, amplitude=8}] (6.3,-0.7)  --  (2.7,-0.7);
			\coordinate [label=center:$B_0$] (B0) at (4.5,-1.5);

			\foreach \from/\to in {b1/a1,b2/a2,bs/as}
	    		\draw (\from) -- (\to);
	    		
	    		\foreach \from/\to in {u/b1,u/b2,u/bs}
	    		\draw (\from) -- (\to);

			\draw (11,4) ellipse [x radius = 3, y radius = 1]; 
			\coordinate [label=center:$Q_1$] (R1) at (11,5.5);
			
			\draw (11,0) ellipse [x radius = 3, y radius = 1];
			\coordinate [label=center:$Q_2$] (R2) at (11,-1.5);				
			
		\end{tikzpicture}
		\label{min_simple_aug_ext_tree_fig}
	\end{figure}

In order to determine whether $G$ contains an augmenting graph $H$ satisfying the above properties, we 
successively consider all triples $(Q_1,Q_2,u)$ such that 
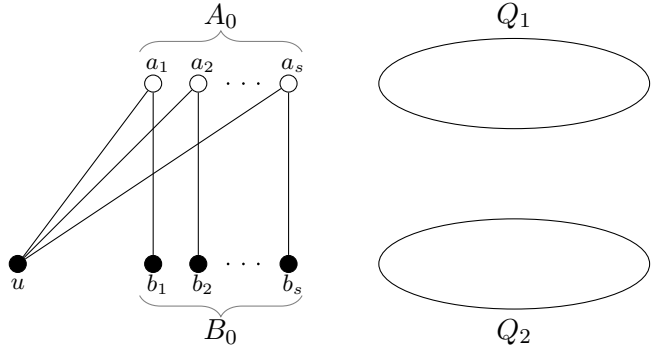
\begin{itemize}
\item $Q_1 \subseteq S$, 
\item $Q_2 \subseteq R = V \setminus S$, 
\item $|Q_1| = |Q_2| \leq 2p$, 
\item $Q_2$ is an independent set and
\item $u$ is a vertex in  $R \setminus Q_2$ with $N(u) \cap Q_2 = \emptyset$. 
\end{itemize}
For each such triple, 
we try to build a copy of $T_k$ centred at $u$. Note that the choice of
	$u$ and $Q_1$ uniquely defines the white part of $T_k$. Namely, 
	$A_0 = N_S(u) \setminus Q_1 = \{ a_1, a_2, \ldots, a_k \}$. 
If $k<p+2$ or $N_S(Q_2) \not\subseteq A_0 \cup Q_1$, then clearly the triple $(Q_1,Q_2,u)$
does not belong to any augmenting graph $H$ satisfying all the properties stated at the beginning of the proof,
in which case we eliminate this triple from further consideration and move to the next one. Otherwise, we check whether
	there is a set of black vertices $B_0 = \{ b_1, b_2, \ldots, b_k \}$ such that:
	\begin{itemize}
		\item $\{u\} \cup B_0 \cup Q_2$ is an independent set;
		\item for $i=1, \ldots, k$, the only white neighbour of $b_i$
		in $S \setminus Q_1$ is $a_i$.
	\end{itemize}
	
	\noindent
	To this end, we consider the following sets for $i=1,\ldots,k$:
	$$
	L_i = \{ v \in R \setminus (Q_2 \cup \{u\}) \ |\  
	N_{S \setminus Q_1} (v)  = \{a_i\} {\text{ and }} 
	v {\text{ has no neighbours in }} \{u\} \cup Q_2 \}.
	$$
If at least one of these sets is empty, then again the triple $(Q_1,Q_2,u)$
is not part of any augmenting graph $H$ we are looking for, and hence we eliminate this triple. 
	Otherwise, for each $i = 1, \ldots, k$ we select any vertex from $L_i$
	as $b_i$ and return the graph $G[Q_1 \cup Q_2 \cup \{u\} \cup A_0 \cup B_0]$.
	It remains to show that $B_0$ is an independent set. Assume for a contradiction that 
$b_i$ is adjacent to $b_j$ for two distinct indices $i,j \in \{1, \ldots, k\}$ and let 
$l,m \in \{1, \ldots, k\}$ be two distinct indices different from $i,j$. Then the set
 $\{ u,a_l, a_m, a_i, b_i, b_j \}$ induces an $S_{1,1,3}$. This contradiction shows that $B_0$ 
 is an independent set and hence $\{u\}\cup A_0\cup B_0 \cup Q_1\cup Q_2$ induces an augmenting 
graph $H$ for $S$ which is a $4p$-extension of a simple tree $T_k$ with $k\ge p+2$. If all triples 
have been examined and eliminated, then no such $H$ exists. 
	
In order to show that the above procedure is polynomial in $n=|V(G)|$, we 
observe that there are $O(n^{4p+1})$ triples $(Q_1,Q_2,u)$
such that $|Q_1|=|Q_2|\le 2p$. Also, it is obvious that for each triple the sets $A_0, L_i$ $(i=1,\ldots,k)$ can be 
constructed in polynomial time. Therefore, for a fixed $p$, the above procedure for detecting 
$4p$-extensions of simple trees takes polynomial time.
\end{proof}

\begin{theorem}
\label{thm:FindMinAugGraph}
For any $p \in \mathbb{N}$, the {\sc maximum independent set} problem can be solved for $(S_{1,1,3},K_{p,p})$-free graphs in polynomial time.
\end{theorem}
\begin{proof}
Let $G$ be an $(S_{1,1,3},K_{p,p})$-free graph and $S$ an {\it arbitrary} independent set in $G$.
If $G$ contains an augmenting path for $S$, such a path can be found by an algorithm proposed in \cite{GHL2006}, 
which works in polynomial time for any graph containing no induced $S_{1,2,3}$.

If $G$ contains a $4p$-extension of a simple tree, such an extension can be found in polynomial time
by Lemma~\ref{lem:FindExtOfSimpleTree}. 

If $G$ contains neither an augmenting path nor an extension of a simple tree for $S$,
then by Theorem~\ref{th:MinAugStruct}, the set $S$ is not maximum if and only if it admits an augmenting graph 
which is $(P_8,T_{p+2},K_{p,p})$-free. By Theorem~\ref{th:MinClasses} there are only finitely many irreducible 
graphs in this set and hence detecting such graphs can be done in polynomial time.  
 
Thus in polynomial time one can determine whether $G$ contains an augmenting graph for $S$. 
Since an augmentation can be applied at most $|V(G)|$ times, we conclude that 
the overall time complexity of finding a maximum independent set in $G$ is polynomial. 
\end{proof}

\section{Conclusion}
\label{sec:con}

In this paper, we proved two main results. First, we identified three minimal infinite classes of augmenting graphs,
and second, we showed that the {\sc maximum independent set} problem restricted to the class of $(S_{1,1,3},K_{p,p})$-free graphs
can be solved in polynomial time. We purposely avoided providing any specific time bound for our solution, because the most 
expensive part of our algorithm deals with finding augmenting graphs from a finite collection of $(P_8,T_{p+2},K_{p,p})$-free graphs.
Estimating the size of a largest graph in this collection involves Ramsey numbers and hence any time bound based on this estimation
is of only theoretical interest. Finding stronger bounds leading to more efficient algorithms for $(S_{1,1,3},K_{p,p})$-free graphs 
is an interesting open problem. 

To state one more open problem, let us observe that our result for $(S_{1,1,3},K_{p,p})$-free graphs  generalizes 
the polynomial-time solution to the problem in the class of claw-free graphs (for each $p\ge 3$). 
This observation and the fact that the problem can be solved for {\it weighted} claw-free graphs \cite{NT01}
raises the following question: is it possible to extend polynomial-time solvability of the problem to {\it weighted} $(S_{1,1,3},K_{p,p})$-free graphs?
We leave this question as an open problem for future research.     

\section*{Acknowledgments}
Research of Konrad Dabrowski was supported by Agence Nationale de la Recherche award ANR-09-EMER-010 and Engineering and Physical Sciences Research Council (EPSRC) award EP/K025090/1.
Vadim Lozin and Viktor Zamaraev acknowledge support from EPSRC grant EP/L020408/1.
Konrad Dabrowski and Vadim Lozin were also supported by the Centre for Discrete Mathematics and its Applications (DIMAP), which was partially funded by EPSRC award EP/D063191/1.
Viktor Zamaraev was partially supported by 
Russian Federation Government grant No.~11.G34.31.0057.
This research was partly carried out while Vadim Lozin was visiting Dominique de Werra at EPFL. The support of EPFL is 
gratefully acknowledged.

\end{document}